\documentclass[notitlepage,floats,aps,nofootinbib,preprintnumbers,superscriptaddress,onecolumn,prl]{revtex4-1}

\usepackage{amsmath}
\usepackage{amsfonts}
\usepackage{amssymb}
\usepackage{amsthm}
\usepackage{graphicx,float}
\usepackage{verbatim}
\usepackage[margin=3.2cm]{geometry}
\usepackage{color}
\newtheorem{theorem}{Theorem}

\def\beq{\begin{equation}\begin{aligned}}
\def\eeq{\end{aligned}\end{equation}}
\def\beqn{\begin{equation*}\begin{aligned}}
\def\eeqn{\end{aligned}\end{equation*}}

\begin{document}

\title{All-Pay Auctions with Different Forfeits}
\author{Benjamin Kang,}
\affiliation{MIT PRIMES, Massachusetts Institute of Technology, Cambridge, MA 02139 }
\author{James Unwin}
\affiliation{Department of Physics,  University of Illinois Chicago, Chicago, IL 60607, USA}
\affiliation{Department of Physics, University of California, Berkeley \& Theoretical Physics Group, LBNL \& Mathematics Sciences Research Institute, Berkeley, CA 94720, USA}

\begin{abstract}
In an auction, each party bids a certain amount, and the one who bids the highest is the winner. Interestingly, auctions can also be used as models for other real-world systems. In an all pay auction all parties must pay a forfeit for bidding. In the most commonly studied all-pay auction, parties forfeit their entire bid, and this has been considered as a model for expenditure on political campaigns. Here, we consider a number of alternative forfeits that might be used as models for different real-world competitions, such as preparing bids for defense or infrastructure contracts.

\vspace{5mm}

\noindent
\textit{Keywords:} All-pay Auction, Auctions, Economic Model, Competition
\end{abstract}

\maketitle



\section{1.~Introduction}

For thousands of years, auctions have been used as a method for selling objects, and four main types of auctions have gained prominence. The first type of auction is the English auction. In this type of auction, the seller continually raises the price of the item until only one person is willing to pay, and the item is sold at this price. A second type of auction is the Dutch auction. In this auction, the seller sets an extremely high price and continually lowers it until a bidder is willing to pay. A third type of auction is the first-price sealed-bid auction. In this type, bidders all bid simultaneously, and the bidder with the highest bid wins and pays that bid. A fourth type of auction is the second-price sealed-bid auction, where bidders simultaneously bid, and the bidder with the highest bid wins but pays the second-highest bid. These four types of auctions have already been analyzed extensively in the literature. 

Although the four auctions previously described are the most popular ones, there are many other variations. Of particular interest are the `war of attrition' and all-pay auction, since in both of these, the bidders who do not win must pay a forfeit. In the former, the winner pays the second-highest bid, while in the latter, the winner pays the highest bid. Notably, conflicts among animals \cite{Bishop} can be represented by the war of attrition. On the other hand all-pay auctions have been used to model the arms race \cite{ONeill} and war outcomes \cite{Hodler,2022LSRS...15..145K}, also  rent-seeking scenarios such as lobbying \cite{Baye} or competition with sunk investments \cite{Siegel,Siegel2009}.

	In a seminal paper, Milgrom \& Weber  \cite{Milgrom}, building on earlier work \cite{Capen,Cassady,Myerson,Vickrey,Wilson}, demonstrated mathematically the equivalence of a number of auction systems under certain assumptions and derived the expected selling prices and optimal bidding strategies. A trio of influential papers, which followed, due to Leininger \& Amann \cite{Amann0,Amann}  and Krishna \& Morgan \cite{Krishna}, expanded on the results of Milgrom \& Weber by calculating the bidder strategies for all-pay auctions and the war of attrition scenario. Additionally, there were also multiple other papers that subsequently analyzed different variants on all-pay auctions, such as the behavior of bidders in an all-pay auction with incomplete information  \cite{Amegashie,Noussair}, and Che \& Gale \cite{Che} studied the relationship between all-pay and first-price auctions. 


In the most commonly studied models of all pay auctions, parties forfeit their entire bid. An important comparison to a real-world situation where both parties pay the full cost of their bids but only the winner profits is expenditure on political campaigns  \cite{Snyder}. One can also explore variants in which the non-winning bidders pay different amounts based on all of the bids. Possibilities which we will discuss here include bidders paying a constant entrance fee or paying a fraction of their bid. These auctions are not as prominent in commercial settings but can be used as models for many other systems, some of which we will highlight.

 Our analysis is related to the contest literature initiated by Baye, Kovenock and de Vries, in particular their complete-information characterisation of the all-pay auction \cite{Baye1996}, its application to lobbying and rent seeking~\cite{Baye}, and their treatment of rank-order contests with spillovers~\cite{Baye2012}. The present paper develops complementary results in a symmetric Bayesian setting with affiliated private signals and linear forfeiture rules. See also Dechenaux, Kovenock \& Sheremeta~\cite{DechenauxSurvey}, which surveys a wide range of contest and all-pay auction models. 
We also highlight Siegel~\cite{Siegel2009} who characterised equilibria in a very general class of complete-information contests with continuous, non-decreasing cost functions and identical prizes. 
This work is also related to more recent literature on contests and all-pay auctions with reimbursements or refunds, including Minchuk~\cite{Minchuk} and Liu \& Yong~\cite{Liu}, which analysed all-pay auctions with winner reimbursement under incomplete information, and Chowdhury et al \cite{Chowdhury}, which examined optimal reimbursement schemes in contests.  These works have typically used general cost-of-effort functions and focus on the choice of reimbursement.  In contrast, here we consider the Milgrom-Weber affiliated-signals framework  \cite{Milgrom} and study how forfeiture rules applied to bids (rather than effort costs) permit a unified analysis of all-pay auctions.

 	 In this paper, we first outline some relevant background in Section 2 and then proceed to extend the results of Krishna and Morgan \cite{Krishna} to all-pay auctions with different forfeits for the losing bidders. In Section 3 we examine auctions with an entrance fee in addition to paying the bid, both when the fee is returned to the winner and when it is not. Section 4 explores the case in which the forfeit function is a constant fraction of the original bid. For these auctions, we derive an expression for the symmetric bidding strategy in each case. Then, between these types of auctions, we compare the revenue made for the seller. In Section 5 we consider the approximate behavior of bidders with an exponential forfeit as the bid grows larger.  Section 6 provides a summary and a discussion of possible future research directions.

Beyond their mathematical tractability, the forfeit functions we study capture a range of economically relevant situations.  
Fractional forfeits arise naturally when losers bear only part of their expenditure. Examples include settings where a fraction of lobbying expenses is covered by third parties, litigation in which some costs are shifted or reimbursed, or contests in which non-winning projects retain some salvage value (for example, technologies that can be repurposed). In such environments, the fraction $\beta$ measures the degree of resource dissipation, and we show how equilibrium behavior and revenues interpolate between first-price and classic all-pay auctions as $\beta$ varies.
Exponential forfeits approximate environments with a sharply convex cost of effort or penalties that escalate steeply with the level of outlays. One can consider regulatory or reputational sanctions that become disproportionately severe beyond certain thresholds, or budget environments in which marginal financing costs increase rapidly with scale. In this case, we show that equilibrium bids become extremely steep and, in the limit, are essentially dictated by the distribution of opponents' signals rather than the underlying valuations.

 \vspace{-3mm}
\section{2.~Background: All-Pay Auctions}
\vspace{-3mm}

	Suppose there are $n$ bidders all competing for a single object. 
	Each bidder $i$ observes a private signal about the object. We collect these in the vector $X = (X_1,\dots,X_n)$, where the components represent the information known by each bidder $i$. 
	 Also introduce $S = (S_1, S_2, \dots S_m)$ which represents additional variables that affect the value of the object but are only known to the seller. Then suppose that there is a nonnegative finite function $u$ such that $u(S, X_i, \{X_j\}_{j \neq i}) = V_i$, giving the value of the object to bidder $i$. The payoff for the winner is $V_i - b$ where $b$ is the price paid.

	An independent private values model is an auction in which bidders are risk-neutral and only know the value of the object to themselves, with values taken from a continuous distribution. Milgrom and Weber  \cite{Milgrom} developed such a model applicable to symmetric auctions. The independent private values case arises as a special case when the seller's signal $S$ is degenerate, and the components of $X$ are independent. In what follows we adopt the more general affiliated-signals environment in order to remain directly comparable to Milgrom and Weber \cite{Milgrom} and Krishna and Morgan \cite{Krishna}, whose results we extend.

Furthermore, denote by $f(S, X_1,\cdots, X_n)$ the joint probability distribution of the random variables, which is symmetric in the last $n$ variables. Note that this function $f$ obeys the `affiliation inequality'  \cite{FKG} given by $f(z \vee z') f(z \wedge z') \geq f(z) f(z')$ where $z \vee z'$ is the component-wise maximum and $z \wedge z'$ is the component-wise minimum. This implies that it is more likely for the variables to be close to each other, rather than farther apart.
Then define $Y_1 = {\rm max}\{X_j\}_{j \neq 1}$ and let $f_{Y_1}(\cdot |x)$ be the conditional density of $Y_1$ with $x=X_1$. We denote the corresponding cumulative distribution as $F_{Y_1}(\cdot |x)$. Note the cumulative distribution of a function $f$ at a point $y$ is defined as the probability that the result is at most $f(y)$ and can be expressed as 
 \beqn F_{Y_1}(y|x) = \int_{-\infty}^y f_{Y_1}(s|x){\rm d}s.\eeqn 

Here we adopt the affiliated private-values framework of Milgrom and Weber~\cite{Milgrom}. 
Each bidder $i$ observes a private signal $X_i$ about the object, and the seller observes additional variables $S$ that affect valuations. 
The joint density $f(S,X_1,\dots,X_n)$ is assumed to satisfy the affiliation condition  \cite{FKG}, which implies positive dependence among bidders' signals. 
Independent private values arise as a special case when $(X_1,\dots,X_n)$ are independent.
This allows us to compare all-pay and first-price auctions within a common framework and to use the monotonicity properties implied by affiliation in our revenue comparisons. 
Independent private values are nested as a special case.

For joint density $f(S,X_1,\dots,X_n)$ satisfying the affiliation condition, Milgrom and Weber showed that the following monotonicity property holds  \cite{Milgrom}:

\vspace{5mm}
 {\bf  Lemma 1} (MW){\bf.} Suppose $f$ is affiliated, then $\frac{F_{Y_1}(x|z)}{f_{Y_1}(x|z)}$ is non-increasing in $z$. \newline

\vspace{-7mm}
 \begin{proof}  The proof is instructive, and thus we reproduce it in Appendix A.
 \end{proof}

	The above model has been used to study both first-price auctions \cite{Milgrom} and all-pay auctions \cite{Krishna}, and will also be the focus of this work. 
	In the classic all-pay auction in which the losers forfeit their bid, one can define a payoff function $W$ of the following form  \cite{Krishna}
	\begin{equation*}\begin{aligned} W_i = \begin{cases}
	V_i - b_i & b_i > {\rm max}_{j \neq i} b_j \\
	-b_i & b_i < {\rm max}_{j \neq i} b_j \\
	\frac{V_i}{\#\{k:b_k = b_i\}} - b_i & b_i = {\rm max}_{j \neq i} b_j
	\end{cases}~.\end{aligned}\end{equation*} 

	Below, we will outline the heuristic for finding the symmetric equilibrium bidding strategy from \cite{Krishna}. Suppose bidders $j \neq 1$ follow the symmetric, increasing equilibrium strategy $\alpha$ and bidder 1 bids $b$ with $X_1 = x$. Denote by  $v$ the expected value $E[\cdot]$ of the object to bidder $1$, defined by 
	\beq \notag
	v(x,y) = E[V_1 | X_1 = x, Y_1 = y].
	\eeq
	Then the expected payoff of bidder $1$, denoted $\Pi(b,x)$, is given by 
\beqn \Pi(b,x) = \int_{-\infty}^{\alpha^{-1}(b)} v(x,y) f_{Y_1} (y | x){\rm d}y - b. \eeqn 
We can obtain the maximum payoff for bidder 1 with respect to the bid by finding when the derivative with respect to $b$ vanishes, which implies \beqn v(x,\alpha^{-1}(b))f_{Y_1}(\alpha^{-1}(b)|x)\frac{1}{\alpha'(\alpha^{-1}(b))} - 1 = 0.\eeqn 
 At symmetric equilibrium, bidder $1$  also follows the bidding strategy $\alpha$, so $\alpha(x) = b$, which gives
\beq \label{eq2} \alpha'(x) = v(x,x)f_{Y_1}(x|x). \eeq
 Integrating this equation, one obtains that the symmetric equilibrium has  the following form
\beqn \alpha(x) = \int_{-\infty}^xv(t,t)f_{Y_1}(t|t){\rm d}t.\eeqn 
	However, this is only a necessary condition for the bidding strategy to be a symmetric equilibrium.
 
	In \cite{Krishna} Krishna \& Morgan the following theorem is proved, which establishes the above as the symmetric equilibrium bidding strategy for bidders in an all-pay auction:

\vspace{2mm}
 \begin{theorem}[KM]
If $ v(x,y)f_{Y_1}(y|x)$ is increasing in x, then the formula for symmetric equilibrium function is given by \beqn\alpha(x) = \int_{-\infty}^xv(t,t)f_{Y_1}(t|t){\rm d}t. \eeqn
 \end{theorem}
 
 \noindent
	Furthermore, Milgrom \& Weber \cite{Milgrom} showed that  the symmetric equilibrium bidding strategy in a first-price auction obeys a similar statement:

\vspace{2mm}
 \begin{theorem}[MW]
 	The function of symmetric equilibrium for a first-price auction is given by \beqn\alpha(x) = \int_{-\infty}^x v(s,s)\frac{f_{Y_1}(s|s)}{F_{Y_1}(s|s)} \exp \left(\int_x^s \frac{f_{Y_1}(t|t)}{F_{Y_1}(t|t)}{\rm d}t \right) {\rm d}s.\eeqn
 \end{theorem}

\noindent
Additionally, from the above theorems, Krishna \& Morgan  \cite{Krishna} proved that if $ v(x,y)f_{Y_1}(y|x)$ is increasing in $x$, then the expected revenue from an all-pay auction is at least as great as that from a first-price auction. These results are important to sellers, as they can help them determine the type of auction to use and the amount they can expect to receive. 
In the remainder of this work, we will explore how the symmetric equilibria $\alpha(x)$ of all pay auctions are impacted by changes to the forfeits that the losing bidders are required to pay by studying changes to the forfeit function~$W$.
 
\vspace{-3mm}
\section{3.~Auctions with Constant Entrance Fees}
\vspace{-3mm}

 In this section, we investigate the effects of introducing a constant entrance fee to an all-pay auction, a possibility mentioned in passing in \cite{Krishna}.
 If a constant entrance fee $c$ is paid by all bidders independently of the auction outcome, the payoff of bidder $i$ can be written as
\beq
  W_i^{\text{fee}} = W_i^{\text{baseline}} - c,
\eeq
where $W_i^{\text{baseline}}$ is the payoff in the classic all-pay auction without a fee.  The fee is introduced as an additive constant and thus leaves the best responses and the symmetric equilibrium bidding function unchanged.  This is the reason that such fees are typically mentioned only in passing in the literature \cite{Krishna}. Thus we focus on the more interesting case in which the entrance fee is rebated to the winning bidder (as might be a model for certain gambling scenarios).

In this case, the expected payoff is given by 
 \beqn W_i = \begin{cases}
 	V_i - b_i & b_i > {\rm max}_{j \neq i} b_j \\
 	-b_i - c & b_i < {\rm max}_{j \neq i} b_j \\
 	\frac{V_i}{\#\{k:b_k = b_i\}} - b_i & b_i = {\rm max}_{j \neq i} b_j
 \end{cases}~.\eeqn
 Suppose bidders $j \neq 1$ follow symmetric increasing equilibrium strategy $\alpha$. The expected payoff of bidder $1$ making a bid $b$ can be expressed as follows
 \beqn \Pi(b,x) = \int_{-\infty}^{\alpha^{-1}(b)} v(x,y) f_{Y_1} (y | x){\rm d}y - b - c(1-F_{Y_1}(\alpha^{-1}(b)|x)),\eeqn
and thus the payoff is maximized for
 \beqn  v(x,\alpha^{-1}(b))f_{Y_1}(\alpha^{-1}(b)|x)\frac{1}{\alpha'(\alpha^{-1}(b))} - 1 + c\cdot f_{Y_1}(\alpha^{-1}(b)|x)\frac{1}{\alpha'(\alpha^{-1}(b))} = 0.
\eeqn 
Similar to the derivation of eq.~\eqref{eq2}, bidder 1 follows the  strategy $ \alpha(x) = b$, thus one has that
 \beqn  
 \alpha'(x) = (v(x,x) + c)f_{Y_1}(x|x)~.
\eeqn
Integrating yields the equilibrium strategy 
  \beqn  
  \alpha(x) &= \int_{-\infty}^x(v(t,t) + c)f_{Y_1}(t|t){\rm d}t
 = \int_{-\infty}^xv(t,t) f_{Y_1}(t|t){\rm d}t + c\int_{-\infty}^x f_{Y_1}(t|t){\rm d}t.
 \eeqn 
	This implies that when the entrance fee is returned to the winner, the bidding strategy changes. This occurs because there is no longer symmetry in the forfeits. Notably, it can be observed that without this symmetry, the optimal bid amount increases as the entrance fee is raised.

To see that $\alpha$ is a symmetric equilibrium, define $\tilde v(x,y) = v(x,y) + c$, which is increasing in $x$ whenever $v$ is. 
The environment with a rebated fee is equivalent to an all-pay auction in which bidder $i$'s value is $\tilde v(X_i,Y_i)$, since the winner recovers the entrance fee. 
Applying Theorem 1 (of Krishna-Morgan)  to $\tilde v$ yields that the symmetric equilibrium bidding strategy is
\beq
  \alpha(x) = \int_{-\infty}^x \tilde v(t,t) f_{Y_1}(t|t)\,{\rm d}t,
\eeq
which coincides with the expression derived above. 
Thus $\alpha$ is indeed the unique symmetric equilibrium in this setting.

Moreover, let us write $\alpha_c$ for the equilibrium bidding function when the entrance fee is $c$, we have
\beq
\alpha_c(x) = \alpha_0(x) + c \int_{-\infty}^x f_{Y_1}(t|t)\,{\rm d}t,
\eeq
It can be seen that $\alpha_c(x)$ is strictly increasing in $c$ for all $x$ in the support of $X_1$. 
Consequently, both the expected bid and the expected total payment of the bidders are increasing in the size of the rebated fee. 
Economically, when the entrance fee is rebated, it becomes part of the stake that is forfeited only upon losing, which strengthens incentives to bid aggressively compared to the baseline all-pay auction without such a rebate.

Writing $F(x) := F_{Y_1}(x\mid x)$, the expected payment of a bidder with signal $x$ in the rebated-fee auction is
\beq
e_c(x) &:= F(x)\,\alpha_c(x) + (1-F(x))(\alpha_c(x) + c)\\
&       = \alpha_c(x) + c(1-F(x)).
\eeq
Substituting $\alpha_c(x) = \alpha_0(x) + c\int_{-\infty}^x f_{Y_1}(t\mid t)\,{\rm d}t$ then yields
\beq
e_c(x) = \alpha_0(x)
+ c\int_{-\infty}^x f_{Y_1}(t\mid t)\,{\rm d}t
+ c(1-F(x)).
\eeq
Observe that $\partial e_c(x)/\partial c > 0$ on the support of $X_1$.  Thus, for every type the expected payment is strictly increasing in the size of the rebated fee. The seller, therefore, always prefers a higher fee: the fee both directly raises payments from losing bidders and indirectly increases bids by strengthening the effective stake that is lost upon losing.

 \vspace{-3mm}
\section{4.~Auctions with Fractional Forfeits}
\vspace{-3mm}

 We next analyze the effects of having the forfeit be a fraction $\beta\in(0,1)$ of each party's bid, therefore
 \beq 
 W_i = \begin{cases}
 	V_i - b_i & b_i > {\rm max}_{j \neq i} b_j \\
 	-\beta b_i  & b_i < {\rm max}_{j \neq i} b_j \\
 	\frac{V_i}{\#\{k:b_k = b_i\}} - b_i & b_i = {\rm max}_{j \neq i} b_j
 \end{cases}~.\eeq
 We next analyse the effects of having the forfeit be a fraction $\beta \in [0,1]$ of each party's bid, so that losers pay $\beta b_i$.
Values of $\beta$ outside this range would correspond to subsidies ($\beta<0$) or penalties exceeding the bid amount ($\beta>1$), which we do not consider here.\footnote{Extensions beyond  $\beta \in [0,1]$ are non-trivial since the Kernel monotonicity arguments inherited from Krishna-Morgan no longer hold; as such, these cases lie beyond the scope of this work.}
Methodologically, our analysis follows the Milgrom-Weber framework and extends the Krishna-Morgan characterisation of symmetric equilibrium bids in all-pay auctions \cite{Krishna}. 
The introduction of a fractional forfeiture parameter $\beta$ yields a class of auctions that contains both first-price auctions ($\beta=0$) and classic all-pay auctions ($\beta=1$) as subcases. 
  
Suppose bidders $j \neq 1$ follow the symmetric increasing equilibrium strategy $\alpha$, then the expected payoff of bidder $1$ is 
 \beq \label{eq3} \Pi(b,x) = \int_{-\infty}^{\alpha^{-1}(b)} v(x,y) f_{Y_1} (y | x){\rm d}y - bF_{Y_1}(\alpha^{-1}(b)|x) - (\beta b )(1-F_{Y_1}(\alpha^{-1}(b)|x)).\eeq 
 It follows that the bid $b$ that maximizes eq.~\eqref{eq3} satisfies the following condition
  \beq\notag & v(x,\alpha^{-1}(b))f_{Y_1}(\alpha^{-1}(b)|x)\frac{1}{\alpha'(\alpha^{-1}(b))} - F_{Y_1}(\alpha^{-1}(b)|x) - bf_{Y_1}(\alpha^{-1}(b)|x)\frac{1}{\alpha'(\alpha^{-1}(b))} \\  
 & \hspace{56mm} - \beta(1 - F_{Y_1}(\alpha^{-1}(b)|x)) + (\beta b )\cdot f_{Y_1}(\alpha^{-1}(b)|x)\frac{1}{\alpha'(\alpha^{-1}(b))} = 0.\eeq 
Multiplying both sides of the above equation by $\alpha'(\alpha^{-1}(b))$ then taking $ \alpha(x) = b$ gives
 \beq 
 \beta \alpha'(x) + (1 - \beta)\alpha'(x)F_{Y_1}(x|x) + (1-\beta)\alpha(x)f_{Y_1}(x|x)= v(x,x) f_{Y_1}(x|x).
 \eeq
	Observe that the above is a first-order differential equation in $\alpha(x)$. Solving this equation, we obtain
 \beqn  
 &\alpha(x) = \int_{-\infty}^x v(s,s) \frac{{\rm d}L(s,x)}{1-\beta}, 
\eeqn 
where 
\beq 
L(s,x) = \exp \left((1-\beta)\int_x^s \frac{f_{Y_1}(t|t)}{\beta + (1-\beta)F_{Y_1}(t|t)}{\rm d}t\right).
\label{asd}  \eeq
	Moreover, this can then be rewritten as follows
 \beq  & \alpha(x) = \int_{-\infty}^x v(s,s)\frac{f_{Y_1}(s|s)}{\beta + (1-\beta)F_{Y_1}(s|s)} \exp \left(-(1-\beta)\int_s^x \frac{f_{Y_1}(t|t)}{\beta + (1-\beta)F_{Y_1}(t|t)}{\rm d}t \right) {\rm d}s. 
 \label{qqq}
 \eeq
	It is interesting that for $\beta = 0$ we obtain the strategy for first-price auctions \cite{Milgrom}, and taking $\beta = 1$ yields the strategy for the classic all-pay auction in \cite{Krishna}.
 \begin{theorem}
 	When $\alpha(x)$ is as defined in eq.~(\ref{qqq}), it is a symmetric equilibrium. 
 \end{theorem}
 \vspace{-5mm}\begin{proof}[Proof sketch]
Equation (\ref{qqq}) can be rewritten by integration by parts as
\beq
  \alpha(x) = \frac{v(x,x) - \int_{-\infty}^x L(s,x)\,{\rm d}t(s)}{1-\beta}.
\eeq
Using the affiliation assumptions, it can be shown that $L(\cdot,x)$ is decreasing in its first argument while $v(x,x)$ is increasing, implying that $\alpha$ is strictly increasing in $x$.
For a bidder with signal $x$, consider a deviation in which they submit a bid $b = \alpha(z)$ corresponding to signal $z$. 
Denote the expected payoff from this deviation $\Pi(\alpha(z),x)$. 
Differentiating $\Pi(\alpha(z),x)$ with respect to $z$ and applying Lemma~1, one finds that $\partial \Pi(\alpha(z),x)/\partial z$ has the same sign as $z-x$. 
Hence $\Pi(\alpha(z),x)$ is maximized at $z=x$, so bidding $b=\alpha(x)$ is optimal given that other bidders follow~$\alpha$.
Possible discontinuities in $\alpha$ can be treated by adapting arguments Milgrom and Weber  \cite{Milgrom}. Further details are given in Appendix~B.
\end{proof}

 \begin{theorem}
 	The expected revenue generated for the seller of an all-pay auction with fractional cost is always less than when $\beta = 1$ if $f(y|x)$ is increasing in $x$.
 \end{theorem}
 \vspace{-5mm}
 \begin{proof}[Proof sketch]

Let $e_\beta(x)$ denote the expected payment of a bidder with signal $x$ under forfeiture parameter $\beta$, and let $\alpha_\beta$ be the corresponding equilibrium bidding function. Substituting $\alpha_\beta$ into the payoff expression shows that
\[
e_\beta(x) = \int_{-\infty}^x v(s,s)\,f_{Y_1}(s\mid s)\,K_\beta(x,s)\,{\rm d}s,
\]
for an explicit kernel $K_\beta(x,s)$.
Using affiliation and Lemma~1, one can show that for $\beta\in[0,1]$ this kernel satisfies $0\leq K_\beta(x,s)\leq 1$ for all $s\leq x$, with $K_1(x,s)\equiv 1$ corresponding to the classic all-pay auction.
It follows that $e_\beta(x) \le e_1(x)$ for every $x$ in the support of $X_1$. 
By integrating this inequality with respect to the distribution of $X_1$ and multiplying by the number of bidders $n$, the result follows.
The full derivation is given in Appendix~C.
\end{proof}

\noindent
Notably, this shows that the expected amount paid by a bidder in an auction where $\beta \leq 1$ is at most the expected price paid by a bidder in the original all-pay auction. Since this holds for each bidder, it follows for the expected revenue earned by the seller as well.
 
 To summarize, in this section we have proved the equilibrium bidding strategy for the all-pay auction with fractional forfeit. Theorem~4 shows that, under affiliation and the monotonicity condition on $f_{Y_1}$, expected revenue is maximised at the full-forfeit case $\beta=1$ when we restrict attention to $\beta\in[0,1]$ (in principle, punitive forfeits $\beta>1$ may alter this conclusion but we do not consider these cases here). For $\beta\in[0,1]$ we also showed that each of these auctions does not generate as much revenue as the all-pay auction with complete bid forfeit.

From an economic perspective, the fractional-forfeit parameter $\beta$ measures the extent to which bids are dissipated when losing. 
When $\beta$ is close to zero, losers recover most of their bids and the auction behaves similarly to a first-price auction, with relatively low rent dissipation. 
As $\beta$ increases toward one, the effective cost of losing rises, bidders internalise the higher prospective loss, and equilibrium bids move toward those of the classic all-pay auction. 
The revenue comparison result shows that, within this linear forfeiture class and under the affiliation assumptions, the seller maximises expected revenue by choosing the full-forfeit case $\beta=1$, so partial reimbursement of bids unambiguously reduces revenue even though it softens the effective cost faced by losers.

\section{5.~Auctions with Exponential Forfeits}

We next consider the interesting case in which the losers must pay an exponentially large forfeit.
An exponential forfeit can be viewed as a stylised representation of environments in which the marginal cost of remaining in a contest increases very rapidly with the level or duration of engagement. 
One example is protracted litigation, where legal fees, management time, and reputational exposure compound over time, so that extending a lawsuit becomes disproportionately more costly at later stages than at earlier ones. 
Similarly, in military or geopolitical conflicts, such as long-running wars, the risks and resource commitments often escalate nonlinearly as parties become more deeply entrenched, with later stages involving substantially higher financial costs and existential risks than initial skirmishes. 
A further illustration is technological races such as the space race, in which maintaining or increasing effort at the technological frontier typically requires orders-of-magnitude increases in expenditure relative to early-stage investments. 
In all these cases, the key economic feature is that the effective cost of `losing late' is much more severe than the cost of `losing early,' and the exponential penalty serves as a tractable reduced-form way to capture such strongly convex escalation in the cost of losing.

The following expected payoff function describes an exponential forfeit
\vspace{2mm}
\begin{equation*}\begin{aligned}
W_i = \begin{cases}
V_i - b_i & \text{if } b_i > {\rm max}_{j \neq i} b_j, \\[2pt]
-\bigl(e^{b_i}-1\bigr)  & \text{if } b_i < {\rm max}_{j \neq i} b_j, \\[2pt]
\displaystyle \frac{V_i}{\#\{k:b_k = b_i\}} - b_i & \text{if } b_i = {\rm max}_{j \neq i} b_j~.
\end{cases}
\end{aligned}\end{equation*}
The normalisation for $b_i < {\rm max}_{j \neq i} b_j,$ ensures that a bidder who submits $b_i=0$ and loses incurs zero monetary loss, so a zero bid corresponds to zero effective outlay. 

This set-up is reminiscent of Baye, Kovenock \& de Vries' analysis of litigation systems \cite{litigation} in which they demonstrated that fee-shifting rules (such as the Quayle system) effectively amplify  legal costs as expenditures rise. Furthermore, their study of the Babylonian bridal auction \cite{herodotus} examined how extreme payment rules can induce highly skewed equilibrium behaviour. Thus the exponential forfeit above can be viewed as a stylised reduced form of such environments, in which losses become disproportionately large at high bids.

Suppose bidders $j \neq 1$ follow a symmetric increasing equilibrium strategy $\alpha$, then the expected payoff of bidder $1$ with bid $b$ and signal $x$ is
\begin{equation*}\begin{aligned}
\Pi(b,x)
= \int_{-\infty}^{\alpha^{-1}(b)} v(x,y) f_{Y_1} (y \mid x)\,{\rm d}y
- b\,F_{Y_1}(\alpha^{-1}(b)\mid x)
- \bigl(e^b-1\bigr)\bigl(1-F_{Y_1}(\alpha^{-1}(b)\mid x)\bigr).
\end{aligned}\end{equation*}
The bid that maximizes the payoff, is given by the critical point with respect to $b$, given by
\beq\notag &v(x,\alpha^{-1}(b))f_{Y_1}(\alpha^{-1}(b)|x)\frac{1}{\alpha'(\alpha^{-1}(b))} - F_{Y_1}(\alpha^{-1}(b)|x) - bf_{Y_1}(\alpha^{-1}(b)|x)\frac{1}{\alpha'(\alpha^{-1}(b))} \\  
&\hspace*{56mm} - e^b(1 - F_{Y_1}(\alpha^{-1}(b)|x)) + e^b\cdot f_{Y_1}(\alpha^{-1}(b)|x)\frac{1}{\alpha'(\alpha^{-1}(b))} = 0.\eeq 
As we derive in Appendix D, 
for $x$ such that $\alpha(x)$ is large, the terms multiplied by $e^{-\alpha(x)}$ are negligible, so to leading order
\beq
\alpha'(x) = \frac{f_{Y_1}(x\mid x)}{1-F_{Y_1}(x\mid x)} + O(e^{-\alpha(x)}).
\label{qwe}
\eeq
In the high-bid region the approximate strategy is
\beq
\alpha(x) \approx \int_{-\infty}^x \frac{f_{Y_1}(t\mid t)}{1-F_{Y_1}(t\mid t)}\,{\rm d}t.
\eeq
In this asymptotic regime the leading-order behaviour of $\alpha$ depends only on the distribution of opponents' signals, while valuations enter only through exponentially small correction terms.

As an explicit example, consider the case of two bidders where we denote  bidder 1's signal by $x$ and bidder 2's signal  by $y$, and suppose 
\beq\notag
f(x, y) = \frac{4}{5}(1 + xy),
\eeq
where $f$ is defined on $[0,1] \times [0,1]$. This implies that $f_{Y_1} (y|x) = \frac{2+2xy}{2 + x}$ and $F_{Y_1}(y|x) = \frac{2y + xy^2}{2 + x}$, thus 
\begin{equation*}\begin{aligned} \alpha(x) = \int_{0}^x \left(\frac{2 + 2t^2}{2 - t - t^3} \right){\rm d}t = \int_{0}^x \left( \frac{1}{1 - t} - \frac{t}{2 + t + t^2}\right) {\rm d}t.
\end{aligned}\end{equation*}
	This function behaves very similarly to $-\ln(1-x)$ since the second term in the integral is negligible. Notice that this function increases slowly at first but then begins to grow increasingly rapidly. This is indicative of the optimal strategy for successful parties in auctions with exponential forfeits.

Economically, the exponential forfeit can be viewed as an extreme form of convex cost of losing, in which the marginal loss associated with a higher bid grows very rapidly. 
In this environment, high types choose bids in a region where the exponential term dominates, and the approximate equilibrium becomes almost entirely driven by beliefs about opponents' signals rather than by the precise level of the underlying valuation $v$. 

Relative to first-price and standard all-pay auctions, the exponential forfeit produces a much steeper bid schedule for high signals and concentrates competition among a narrow set of top types. This illustrates how sufficiently convex forfeiture rules can generate near ``winner-take-all'' behaviour with very strong incentives at the top of the signal distribution.
At such bid levels, the trade-off governing the marginal bid decision is essentially between a slight reduction in the probability of incurring a very large loss and a slight increase in that loss, so the exact valuation $v(x,x)$ affects the optimal bid only at lower order.

The above properties hold in the high-bid region where the exponential penalty dominates; outside this regime, the dependence of equilibrium bids on valuations is more complex and is not characterised here. Valuations remain important for intermediate types and for determining which types ever enter the high-bid region, but the marginal incentives of sufficiently high types are dominated by the structure of the loss rather than by the gain from winning.

\section{6.~Concluding Remarks}
\vspace{-3mm}

In this work we have investigated the impact of changing the forfeit function in all pay auctions. We highlighted that the addition of a constant entrance fee does not affect the bidding strategy unless the fee is returned to the winner. When the forfeit is instead a fraction of the bid, we showed that the revenue generated by the seller is increasing with the fraction.  Lastly, when the forfeit is exponential, the bidding strategy quickly approaches infinity, and it was argued that successful bidders will be those who bid significantly more than the typical bid.

The analysis above permits a simple comparative-statics perspective that connects our results to the existing literature.  In the fractional-forfeit model, the parameter $\beta \in [0,1]$ interpolates smoothly between first-price auctions ($\beta=0$) and the classic all-pay auction ($\beta=1$).  For a fixed information structure, the equilibrium bid function $\alpha_\beta(x)$ and expected payments $e_\beta(x)$ move monotonically with $\beta$, and the revenue comparison theorem shows that seller revenue is maximised at the full-forfeit benchmark studied by Krishna and Morgan~\cite{Krishna}.  This provides a unified framework in which the revenue properties of first-price, all-pay, and intermediate forfeiture rules can be compared.

Contests with features closely related to the forfeiture rules we study arise in various applied settings, including lobbying and rent-seeking \cite{Baye}, electoral campaigns \cite{Snyder}, and research tournaments with sunk investments \cite{Siegel}.  Understanding how alternative forfeiture rules influence equilibrium behaviour, and revenue is therefore relevant both for theoretical auction design and for the interpretation of such real-world contests.

It would be interesting to consider all-pay auctions with a wider range of forfeit functions, such as logarithmic, polynomial, or constant functions. Likely, this is most readily implemented by considering fractional forfeits for all parties, with the forfeit value differing depending on the ranking of each party's bid, with the forfeits following a specified distribution. One might also explore the difference in results if the bidders are risk-averse rather than risk-neutral, or the effects of multiple prizes on the results, generalizing \cite{Barut} to different forfeit functions. 
Furthermore, we have restricted our attention to symmetric equilibria in symmetric environments.  A rich set of asymmetric equilibria can arise in complete-information all-pay contests with spillovers and heterogeneity, cf.~\cite{Siegel,Baye2012}. Extending our analysis of forfeiture rules to such asymmetric settings would be an interesting direction for future work.

Among the possible extensions above, two seem particularly connected to our analysis.  First, allowing for risk-averse bidders within the fractional-forfeit framework would reveal how changes in forfeiture rules interact with curvature of utility, and could be analysed by adapting the differential-equation methods developed here.  Second, introducing rank-dependent forfeits in multiple-prize contests, in the spirit of Barut~\cite{Barut}, would permit a systematic comparison of how different patterns of partial cost reimbursement shape rent dissipation across the full distribution of prizes.

\vspace{2mm}
{\bf Acknowledgements.} This research was undertaken as part of the MIT-PRIMES program.  JU is supported by NSF grant PHY-2209998 and is grateful for support from NSF grant  DMS-1440140 while in residence at the Mathematics Sciences Research Institute, Berkeley during Fall 2019.  

\appendix
\section{Appendix A:  Reproduction of the proof of Lemma 1}
\label{ApA}

{\bf  Lemma 1} (MW){\bf.} Suppose $f$ is affiliated, then $\frac{F_{Y_1}(x|z)}{f_{Y_1}(x|z)}$ is non-increasing in $z$. \newline

\vspace{-7mm} \begin{proof}
 By the affiliation inequality, for $\alpha \leq x$ and $z' \leq z$, 
 \beq \notag
 f_{Y_1}(\alpha|z)f_{Y_1}(x|z') \leq f_{Y_1}(\alpha|z')f_{Y_1}(x|z) \implies \frac{f_{Y_1}(\alpha|z)}{f_{Y_1}(x|z)} \leq \frac{f_{Y_1}(\alpha|z')}{f_{Y_1}(x|z')}. 
 \eeq
Integrating  both sides with respect to $\alpha$ from $-\infty$ to $x$ gives
 \beqn
 \frac{F_{Y_1}(x|z)}{f_{Y_1}(x|z)} \leq \frac{F_{Y_1}(x|z')}{f_{Y_1}(x|z')}.
 \eeqn
 \end{proof}

\vspace{-10mm}
\section{Appendix B: Proof of Theorem 3}
\label{ApB}

\setcounter{theorem}{2}
 \begin{theorem}
 	When $\alpha(x)$ is as defined in eq.~(\ref{qqq}), it is a symmetric equilibrium. 
 \end{theorem}
  \vspace{-5mm}
 \begin{proof} 
 Let $t(x) = v(x,x)$, then by integration by parts $\alpha(x)$ can be rewritten as
 	\beqn
		\alpha(x) = \frac{v(x,x) - \int_{-\infty}^x L(s,x){\rm d}t(s)}{1-\beta}.
	\eeqn
	  We adapt an argument in \cite{Milgrom}\footnote{Specifically, Theorem 14 of \cite{Milgrom} which studied a case corresponding to the $\beta=0$ case of our generalized set-up.} in which it was observed that $L(\alpha | x)$ is decreasing in $x$ and $v(x,x)$ is increasing, so $\alpha(x)$ will grow with increasing $x$. First suppose $\alpha(x)$ is continuous, then we can assume $\alpha(x)$ is differentiable without loss of generality by monotonically rescaling bidder estimates. To show that $\alpha(x)$ is the optimal bid, we find the maximum of the payoff function
 \beqn 
 \frac{\partial}{\partial b}\Pi(\alpha(z), x) = \frac{f_{Y_1}(z|x)}{\alpha'(z)}\left(v(x,z)-(1-\beta)\alpha(z) - \alpha'(z)\frac{ (1-\beta)F_{Y_1}(z|x)}{f_{Y_1}(z|x)}\right) - \beta.
 \eeqn
Applying Lemma 1 and the fact that $v(x,z)$ is increasing it follows that $\frac{\partial}{\partial b}\Pi(\alpha(z), x)$ has the same sign as $z-x$, this implies that $\Pi(\alpha(z), x)$ is maximized for $z = x$.

It remains to check cases in which  $\alpha$ is discontinuous at some point $x$. In this case for any positive $\epsilon$, the following expression is infinite
 \beqn 
 \int_x^{x + \epsilon}  \frac{(1-\beta)f_{Y_1}(s|s)}{\beta + (1-\beta)F_{Y_1}(s|s)} \sim \infty .
\eeqn
It then follows that
\beqn 
 \int_x^{x + \epsilon}  \frac{(1-\beta)f_{Y_1}(s|s)}{\beta + (1-\beta)F_{Y_1}(s|s)} 
 &\leq \int_x^{x + \epsilon}  \frac{f_{Y_1}(s|s)}{\beta F_{Y_1}(s|s) + (1-\beta) F_{Y_1}(s|s)} \\
 &= \int_x^{x + \epsilon}  \frac{f_{Y_1}(s|s)}{F_{Y_1}(s|s)} \\ 
 &\leq \int_x^{x + \epsilon}  \frac{f_{Y_1}(s|x+\epsilon)}{F_{Y_1}(s|x + \epsilon)} \\
 & = \ln(F_{Y_1}(x + \epsilon|x+\epsilon)) - \ln(F_{Y_1}(x|x+\epsilon)).
 \eeqn
For the last expression to be infinite, it is required that $F_{Y_1}(x|x + \epsilon) = 0$, which is a statement proved in Theorem 14 of \cite{Milgrom}. Therefore,  $\alpha(x)$ as given in eq.~(\ref{qqq}) is an equilibrium for this auction.
 \end{proof}

\section{Appendix C: Proof of Theorem 4}
\label{ApC}

\begin{theorem}
 	The expected revenue generated for the seller of an all-pay auction with fractional cost is always less than when $\beta = 1$ if $f(y|x)$ is increasing in $x$.
 \end{theorem}
  \begin{proof} Let $\alpha_\beta(x)$ be the equilibrium bid for a specific value of $\beta$. Notice that the expected payment of a bidder is 
 \beqn  e_\beta(x) & 
 = (F_{Y_1}(x|x) + \beta(1-F_{Y_1}(x|x)))\alpha_\beta(x) \\[5pt]
 &= \int_{-\infty}^x v(s,s)f_{Y_1}(s|s)\frac{\beta + (1-\beta)F_{Y_1}(x|x)}{\beta + (1-\beta)F_{Y_1}(s|s)} \exp \left(-\int_x^s \frac{(1-\beta)f_{Y_1}(t|t)}{\beta + (1-\beta)F_{Y_1}(t|t)}{\rm d}t \right) {\rm d}s.\eeqn
 Since $f_{Y_1}(y|x)$ is increasing in $x$, it follows that $\beta/f_{Y_1}(y|x)$ is decreasing in $x$. Combined with Lemma 1, this implies that  $\frac{f_{Y_1}(y|x)}{\beta + (1-\beta)F_{Y_1}(y|x)}$ is increasing in $x$ and therefore
  \beqn  - \int_s^x \frac{(1-\beta)f_{Y_1}(t|t)}{\beta + (1-\beta)F_{Y_1}(t|t)}{\rm d}t 
  & \leq -\int_s^x \frac{(1-\beta)f_{Y_1}(t|s)}{\beta + (1-\beta)F_{Y_1}(t|s)}{\rm d}t \\ 
& =  \ln(\beta + (1-\beta)F_{Y_1}(s|s)) - \ln(\beta + (1-\beta)F_{Y_1}(x|s)) \\ 
 &\leq \ln(\beta + (1-\beta)F_{Y_1}(s|s)) - \ln(\beta + (1-\beta)F_{Y_1}(x|x))~, \eeqn
where the last inequality comes from the fact that $F_{Y_1}(y|x)$ is non-increasing in $x$.
 It follows that 
 \beqn   e_\beta(x) & \leq \int_{-\infty}^x v(s,s)f_{Y_1}(s|s)\frac{\beta + (1-\beta)F_{Y_1}(x|x)}{\beta + (1-\beta)F_{Y_1}(s|s)}\exp\left(\ln\left(\frac{\beta + (1-\beta)F_{Y_1}(s|s)}{\beta + (1-\beta)F_{Y_1}(x|x)} \right)\right) \\
 & \leq \int_{-\infty}^x v(s,s)f_{Y_1}(s|s) = e_1(x).
\eeqn 
\end{proof}

\vspace{-10mm}
\section{Appendix D: Derivation of $\alpha'(x)$ in exponential case}
\label{ApD}

The bid that maximizes the payoff, is given by the critical point with respect to $b$, given by
\beq\notag &v(x,\alpha^{-1}(b))f_{Y_1}(\alpha^{-1}(b)|x)\frac{1}{\alpha'(\alpha^{-1}(b))} - F_{Y_1}(\alpha^{-1}(b)|x) - bf_{Y_1}(\alpha^{-1}(b)|x)\frac{1}{\alpha'(\alpha^{-1}(b))} \\  
&\hspace*{56mm} - e^b(1 - F_{Y_1}(\alpha^{-1}(b)|x)) + e^b\cdot f_{Y_1}(\alpha^{-1}(b)|x)\frac{1}{\alpha'(\alpha^{-1}(b))} = 0.\eeq 
	Multiplying both sides of the above equation  by $\alpha'(\alpha^{-1}(b))$ and evaluating for $\alpha(x) = b$ gives 
\begin{equation*}\begin{aligned}
e^{\alpha(x)} \alpha'(x) + (1 - e^{\alpha(x)})\alpha'(x)F_{Y_1}(x|x) + (\alpha(x) - e^{\alpha(x)})f_{Y_1}(x|x)= v(x,x) f_{Y_1}(x|x),
\end{aligned}\end{equation*}
which we rearrange to find an expression for $\alpha'$
\begin{equation*}\begin{aligned}
\alpha'(x) = \frac{(v(x, x) + e^{\alpha(x)} - \alpha(x))f_{Y_1}(x|x)}{e^{\alpha(x)} + (1-e^{\alpha(x)})F_{Y_1}(x|x)}.
\end{aligned}\end{equation*}
The differential equation which arises cannot be  solved analytically. However, for large bids one may consider the approximate behavior for in the large $b$ limit, which leads to the form stated in eq.~(\ref{qwe}) of the main text, namely
\begin{equation*}\begin{aligned} 
\alpha'(x) \approx \frac{f_{Y_1}(x|x)}{1 -F_{Y_1}(x|x)} + O(e^{-\alpha(x)}).
\end{aligned}\end{equation*}

\begingroup
\renewcommand{\section}[2]{}%

\endgroup

\end{document}